\documentclass[letterpaper, 10 pt, conference]{IEEEtran}

\IEEEoverridecommandlockouts                              
\usepackage[letterpaper, margin=1in]{geometry}
\usepackage{graphics} 
\usepackage{mathptmx} 
\usepackage{times} 
\usepackage{amsmath} 
\usepackage{amssymb}  
\usepackage{mathtools}
\usepackage{bigints}
\usepackage{cite}
\usepackage{url}

\usepackage[english]{babel}

\newtheorem{theorem}{Theorem}[section]

\newtheorem{prop}[theorem]{Proposition}

\newcommand{\QED}{\null\nobreak\hfill\ensuremath{\blacksquare}}

\newenvironment{proof}[1][\indent \textit{Proof:}]{\begin{trivlist}
\item[\hskip \labelsep #1]}{\end{trivlist}}

\title{\LARGE \bf
Random-Sampling Monte-Carlo Tree Search Methods for Cost Approximation in Long-Horizon Optimal Control}

\author{Shankarachary Ragi, \emph{IEEE Senior Member} and Hans D.\ Mittelmann
\thanks{This work was supported in part by Air Force Office of Scientific Research under grant FA9550-19-1-0070. This paper was presented in part at The 4th IEEE Conference on Control Technology and Applications 2020 \cite{Ragi_CCTA}.}
\thanks{Shankarachary Ragi (corresponding author) is with Department of Electrical Engineering, South Dakota School of Mines and Technology, Rapid City, SD 57701, USA
        {\tt\small shankarachary.ragi@sdsmt.edu}}%
\thanks{Hans D.\ Mittelmann is with the School of Mathematical and Statistical Sciences, Arizona State University, Tempe, AZ 85281, USA
        {\tt\small nlanchie@asu.edu, mittelmann@asu.edu}}
}

\begin{document}

\maketitle

\begin{abstract}

In this paper, we develop Monte-Carlo based heuristic approaches to approximate the objective function in long horizon optimal control problems. In these approaches, to approximate the expectation operator in the objective function, we evolve the system state over multiple trajectories into the future while sampling the noise disturbances at each time-step, and find the average (or weighted average) of the costs along all the trajectories. We call these methods \emph{random sampling - multipath hypothesis propagation} or RS-MHP. These methods (or variants) exist in the literature; however, the literature lacks results on how well these approximation strategies converge. This paper fills this knowledge gap to a certain extent. We derive convergence results for the cost approximation error from the RS-MHP methods and discuss their convergence (in probability) as the sample size increases. We consider two case studies to demonstrate the effectiveness of our methods - a) linear quadratic control problem; b) UAV path optimization problem. 
\end{abstract}


\section{Introduction}
Long-horizon optimal control problems appear naturally in robotics, advanced manufacturing, and economics, especially in applications requiring decision making in stochastic environments. Often these problems are solved via dynamic programming (DP) formulation \cite{Chong2009}. DP problems are notorious for their computational complexity, and require approximation approaches to make them tractable. A plethora of approximation techniques called \emph{approximate dynamic programs} (ADPs) exist in the literature to solve these problems approximately. Some of the commonly used ADPs include \emph{policy rollout} \cite{Bertsekas_PolicyRoll}, \emph{hindsight optimization} \cite{Hindsight1,Hindsight2}, etc. A survey of the ADP approaches can be found in \cite{Chong2009}. Feature-based techniques and deep learning methods are gaining importance in the development of ADP approaches as discussed in \cite{Bertsekas_deeplearn}. These approximation techniques have been successfully adopted to solve real-time problems such as a UAV guidance control problem in \cite{Ragi_UAVpathPlanning, Ragi_DynUAV, Ragi_ACC2017}. Certain ADP approaches, especially the methods based on approximation in value space, require numerical approximation of the expectation in the objective function \cite{Ragi_UAVpathPlanning}. In this study, our objective is to develop Monte-Carlo-based approaches to approximate the expectation in the objective function in the long (but finite) horizon optimal control problems, and study their convergence. A preliminary version of the parts of this paper were published as \cite{Ragi_CCTA}. This paper differs from the conference paper \cite{Ragi_CCTA} in the following ways: 1) we include detailed proofs omitted in the conference version; 2) we derive new convergence results and proofs in Section~\ref{sec:nonoverlap}; 3) we implement our methods for a new case study - UAV path optimization problem.      

\subsection{Preliminaries}
A long horizon optimal control problem is described as follows. Let $x_k$ be the state vector for a system at time $k$, which evolves according to a discrete stochastic process as follows: 
\begin{equation}\label{eq:lhc_cost}
x_{k+1} = f(x_k,u_k,w_k)
\end{equation} 
where $f(\cdot)$ represents the state-transition mapping, $u_k$ is the control vector, and $w_k$ random disturbance. Let $g(x_k,u_k)$ represent the cost (a real value) of being in state $x_k$ and performing action $u_k$. The functions $f$ and $g$ are independent of $k$ in our study, but can generally depend on $k$. The goal is to optimize the control vectors $u_k, k=0,\ldots,H-1$ such that the expected cumulative cost is minimized, i.e., the goal leads to solving the following optimization problem
\begin{equation}\label{eq:lhc_optim}
\min_{u_k, k=0,\ldots, H-1} \,\, \text{E} \left[\sum_{k=0}^{H-1} g(x_k,u_k)\right],
\end{equation} 
where $H$ is the length of the planning horizon. Let $x_0$ be the initial state and according to the dynamic programming formulation the optimal cost function is given by
\begin{equation}\label{eq:lhc_DP}
J_0^*(x_0) = \min_{u_0} \,\, \text{E} \left[ g(x_0,u_0) + J_1^*(x_1) \right],
\end{equation}     
where $J_1^*$ represents the optimal cost-to-go from time $k=1$, and $x_1 = f(x_0,u_0,w_0)$. In this study, \emph{long horizon} refers to the condition that $H$ is sufficiently large that the optimal policy is approximately \emph{stationary} (independent of $k$). Solving the above optimization problem is not tractable mainly due to two reasons: the expectation $E[\cdot]$ and the optimal cost-to-go $J_1^*$ are hard to evaluate and are usually approximated by numerical methods or ADP approaches. 

An ADP approach called \emph{nominal belief-state optimization} (NBO) \cite{Ragi_UAVpathPlanning,Chong_NBO} was developed primarily to approximate the above expectation. In NBO, the expectation is replaced by a sample state trajectory generated with an assumption that the future noise variables in the system take so called nominal or mean values, thus making the above objective function deterministic. The NBO method was developed to solve a UAV path optimization problem, which was posed as a \emph{partially observable Markov decision process} (POMDP). POMDP generalizes the long horizon optimal control problem described in Eq.~\ref{eq:lhc_optim} in that the system state is assumed to be ``partially'' observable, which is inferred via using noisy observations and Bayes rules. Although the performance of the NBO approach was satisfactory, in that it allowed to obtain reasonably optimal control commands for the UAVs, it ignored the uncertainty due to noise disturbances thus leading to inaccurate evaluation of the objective function. To address this challenge, certain methods exist in the literature usually referred to as Monte-Carlo Tree Search (MCTS) methods as surveyed in \cite{Browne_mcts}. 

Inspired from the NBO method and MCTS methods, we develop a new MCTS method called \emph{random sampling - multipath hypothesis propagation} (RS-MHP) and derive convergence results. In this study, we use the NBO approach as a benchmark for performance assessment since RS-MHP builds on the NBO approach. 

\section{Random Sampling Multipath Hypothesis Propagation (RS-MHP)}
In the NBO method, the expectation is replaced by a sample trajectory of the states (as opposed to random states) generated by 
\begin{equation}
\tilde{x}_{k+1} = f(\tilde{x}_k,u_k,\bar{w}_k), \,\, k=0,\ldots
\end{equation}         
where $\tilde{x}_0 = x_0$ (initial state or current state), and $\bar{w}_k$ is the mean of the random variable $w_k$. Thus, the long horizon optimal control problem, with NBO approximation, reduces to 
\begin{equation}\label{eq:lhc_optim_NBO}
\min_{u_k} \,\, \sum_{k=0}^{H-1} g(\tilde{x}_k,u_k).
\end{equation} 
The above reduced problem, without the need for evaluating the expectation, can significantly reduce the computational burden in solving the long horizon control problems. However, the downside with this approach is it completely ignores the uncertainty in the state evolution, and may generate severely sub-optimal controls. To overcome this trivialization, we develop a Monte-Carlo approach to approximate the expectation described as follows. We will follow the tree-like sampling approach as in Figure~\ref{fig:tree_models}(a). For time step $k=1$, we sample the probability distribution of the noise disturbance $N$ times to generate the samples $w_0^i$ with corresponding probability $p_{0}^i$, $i=1,\ldots,N$. Using these, we generate $N$ sample states at $k=1$ generated according to
\begin{equation}
x_1^i = f(x_0,u_0,w_0^i), \,\, \forall i.    
\end{equation} 
We repeat this sampling approach for time $k=2$, i.e., we generate $N$ noise samples $w_1^i$ with corresponding probability $p_{1}^i$, $i=1,\ldots,N$. Using these noise samples and the sample states from the previous time step, we generate $N^2$ sample states at $k=2$ according to
\begin{equation}x_2^{i,j} = f(x_1^i,u_1,w_1^j), \,\, \forall i,j. \end{equation} 

\begin{figure}[t]
\centering
    \includegraphics[width= \columnwidth, trim = 200 260 150 80,clip]{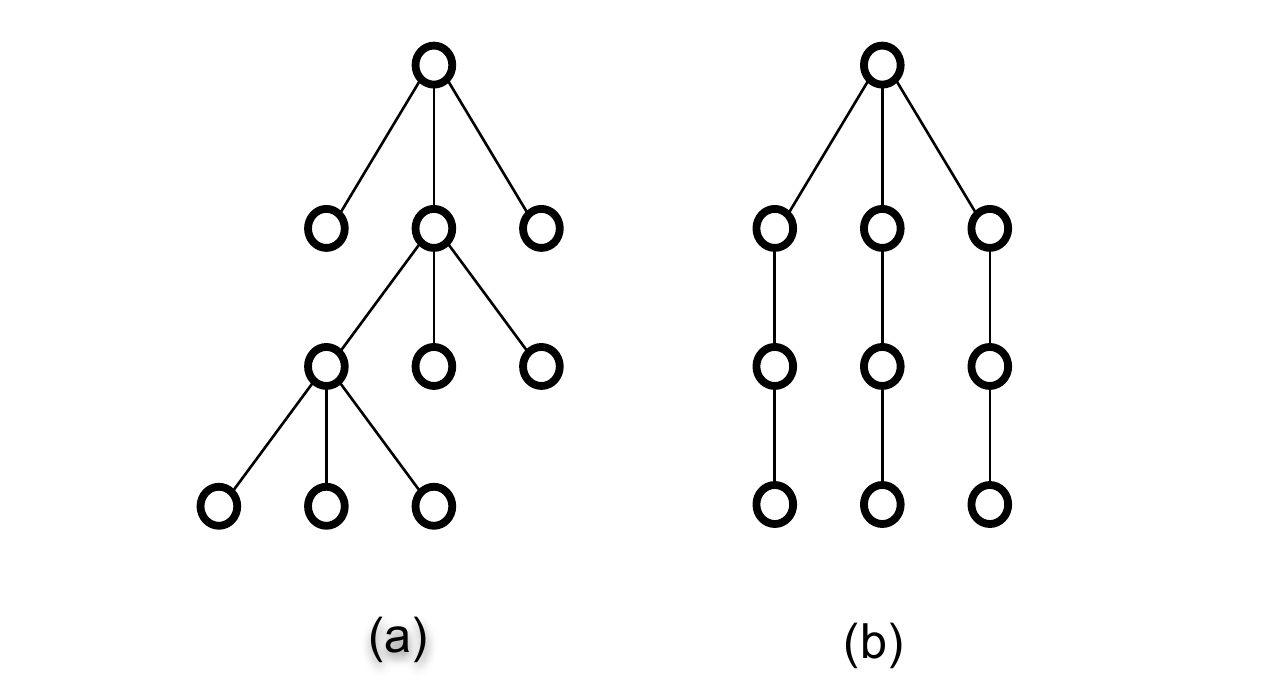}
    \caption{State trajectory sampling models: (a) tree branching model, (b) non-overlapping branching model.}
    \label{fig:tree_models}
\end{figure}

We repeat the above sampling procedure until the last time step $k=H-1$ to generate $N^{H-1}$ possible state evolution trajectories using $N$ noise samples generated in each time step as depicted in Figure~\ref{fig:tree_models}(a). Sampling approach in Figure~\ref{fig:tree_models}(b) will be discussed later.  

One can now replace the expectation in Eq.~\ref{eq:lhc_optim} with the weighted average of the cumulative cost corresponding to each state evolution trajectory, where the weights are the probabilities or likeliness of the trajectories. Clearly, the number of possible state trajectories grow exponentially with the horizon length $H$. Although this approach is not novel as many such methods exist in the literature often classified as Monte-Carlo Tree Search methods, our study is focused on deriving convergence results of RS-MHP approaches. 

To avoid the exponential growth in our RS-MHP approach, at each time step we retain only $M$ sample states and prune the remaining states, and if the number of sample states at a given time instance is less than or equal to $M$, we do not perform pruning. For pruning, at each time $k$, we rank the state trajectories up to time $k$ according to their likeliness (obtained by multiplying the probabilities of all the noise samples that generated the trajectory) and retain the top $M$ trajectories with highest likeliness and prune the rest. With this procedure, at $k=H-1$, there would be only $M$ state trajectories. With pruning, the number of trajectories remains a constant irrespective of the time horizon length. An illustration of the above RS-MHP approach is shown in Figure~\ref{fig:noise_sampling} along with the NBO approach. Here, we consider pruning based on likeliness of the state trajectories as the costs from these trajectories have higher contribution in the cost function in Eq.~\ref{eq:lhc_cost} than the less likely trajectories. We will consider other pruning strategies to further improve the approximation error in our future study.     

\begin{figure}[t]
\centering
    \includegraphics[width= 0.8\columnwidth, trim = 0 150 380 100,clip]{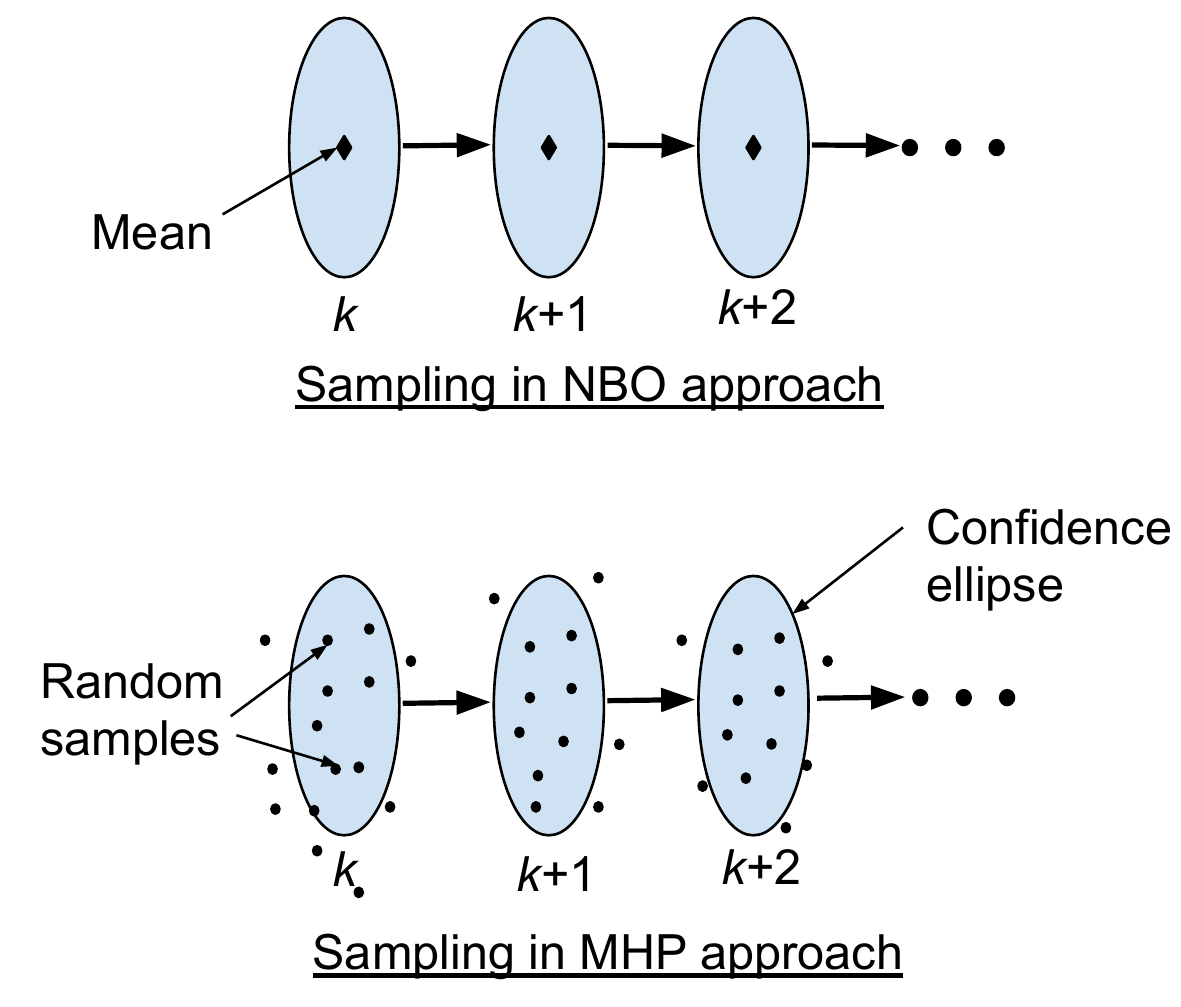}
    \caption{Sampling probability distributions of noise variables: NBO vs.\ MHP.}
    \label{fig:noise_sampling}
\end{figure}

Let $i=1,\ldots,M$ represent the indices of the $M$ distinct state trajectories with $q_1,q_2,\ldots$ being their likeliness index evaluated using the probabilities of the noise samples that generate the trajectory $i$ over time. Let $J$ represent the actual objective function as described below 
\begin{equation}
    J = \text{E} \left[\sum_{k=0}^{H-1} g(x_k,u_k) \right].
\end{equation}

We can now approximate the objective function $J$ in four possible ways as described below (assuming $N>M$). Let $x_k^i$ represent the state at time $k$ in the $i$th state trajectory. 
\begin{itemize}
    \item[(I)] \emph{Sample Averaging}. We can simply approximate the expectation with an average over all possible trajectories as follows:   
    \begin{equation}
    \begin{aligned}
    &\text{No pruning: }  J \approx \tilde{J}_{NP} =   \frac{1}{N^{H-1}}\sum_{i=1}^{N^{H-1}} \left( \sum_{k=0}^{H-1} g(x_k^i,u_k) \right)\\
    &\text{With pruning: }  J \approx \tilde{J}_P = \frac{1}{M}\sum_{i=1}^M \left( \sum_{k=0}^{H-1} g(x_k^i,u_k) \right)\\
    \end{aligned}
    \end{equation}
    
    \item[(II)] \emph{Weighted Sample Averaging}. We can also approximate the expectation with a weighted average with weights being the normalized likeliness indices of the state trajectories given by $q_i, i = 1,\ldots$ (and $\bar{q}_i$ in the pruned case) as follows:   
      \begin{equation}
    \begin{aligned}
    &\text{No pruning: }  J \approx \bar{J}_{NP} = \frac{1}{N^{H-1}} \sum_{i=1}^{N^{H-1}} q_i \left( \sum_{k=0}^{H-1} g(x_k^i,u_k) \right)\\
    &\text{With pruning: }  J \approx \bar{J}_P = \frac{1}{M}\sum_{i=1}^M \bar{q}_i \left( \sum_{k=0}^{H-1} g(x_k^i,u_k) \right).\\
    \end{aligned}
    \end{equation}
    where $\sum_{i=1}^{N^{H-1}} q_i = N^{H-1}$ and $\sum_{i=1}^{M} q_i = M$.
\end{itemize}

For a given sequence of control decisions $u_0,u_1,\ldots$, let $g_i$ denote the cost of the $i$th trajectory given by 
\begin{equation}g_i = \sum_{k=0}^{H-1} g(x_k^i,u_k).\end{equation}
Clearly, $g_1,g_2,\ldots$ are identically distributed random variables, but are dependent due to the overlapping state trajectories in the tree-like sampling approach in Figure~\ref{fig:tree_models}(a), where $\text{E}[g_i] = J,\, \forall i$. 

The below result suggests that with sufficient number of sample state trajectories (large $N$), the approximation error in $\tilde{J}_{NP}$ becomes small enough to ignore. 
\begin{prop}\label{lemma:stronglaw}
For any given sequence of actions $u_0, u_l, \ldots$, if the random variables $g_1, g_2, \ldots$ have finite variances, $\tilde{J}_{NP}$ converges to $J$ in probability.        
\end{prop}
\begin{proof}
From Ex.~254 in \cite{Cacoullos_Ex_in_Prob}, we know that $\tilde{J}_{NP} \xrightarrow{\text{P}} J$ if 
\begin{equation}
\lim_{|i-j|\rightarrow \infty} \text{Cov}(g_i,g_j) = 0,
\end{equation} 
where Cov() represents covariance. Suppose, the sequence $g_1, g_2, \ldots$ is arranged such that $g_1$ represents the cost for the left-most branch in Figure~\ref{fig:tree_models}(a), and $g_2$ representing the second branch from the left, and so on. Clearly, the first $g_1, g_2, \ldots, g_N$ are dependent random variables as they share the same parent node, whereas the next $N$ terms $g_{N+1}, g_{N+2},\ldots,g_{2N}$, although dependent among themselves, are independent of the previous $N$ terms (as these branches evolve from a separate parent node), and so on. Thus, $\text{Cov}(g_i,g_j)=0$ if $|i-j|>N$, which implies $\lim_{|i-j|\rightarrow \infty} \text{Cov}(g_i,g_j) = 0$.\QED
\end{proof} 


Furthermore, we can apply similar arguments to prove the convergence of  $\bar{J}_{NP}$ in probability. 

\begin{prop}\label{lemma:weighted_conv}
For a given sequence of actions $u_0, \ldots, u_{H-1}$, if $g_1, g_2, \ldots$ have finite variances, then $\bar{J}_{NP}$ converges to $J$ in probability.     
\end{prop}

\begin{proof}
From \cite{Nasrollah_conv_weigthavg}, we know that if $\tilde{J}_{NP} \xrightarrow{\text{P}} J$ (which is true as shown in Proposition~\ref{lemma:stronglaw}), and if the weights $q_1, q_2, \ldots$ are monotonically decreasing, then $\bar{J}_{NP} \xrightarrow{\text{P}} J$.  Without loss of generality, we can arrange the trajectory costs $g_i$ such that their likeliness indices are monotonically decreasing, i.e., $q_1\geq q_2 \geq q_3 \geq \ldots$, which completes the proof. \QED   

%
\end{proof}


\subsection{Non-overlapping State Trajectories or Tree Branches}\label{sec:nonoverlap}
Suppose the state sample trajectories are generated independently of each other, where the state trajectories do not share any common state samples as depicted in Figure~\ref{fig:tree_models}b. In this new sampling approach, given $u_0,u_1,\ldots$ are the control decisions over the planning horizon, let $p_i$ represent the cost associated with the $i$th state trajectory. We can approximate the LHC objective function as follows:
\begin{equation}
\begin{aligned}
\bar{J}_N & =    \frac{1}{N}\sum_{i=1}^{N} p_i \\
\tilde{J}_N & =    \frac{1}{N}\sum_{i=1}^{N}q_i p_i,
\end{aligned}
\end{equation}
where $q_i$ represents the likeliness index of the $i$th trajectory and $\sum_{i} q_i = N$. From propositions \ref{lemma:stronglaw} and \ref{lemma:weighted_conv}, we can verify that $\bar{J}_{N} \xrightarrow{\text{P}} J$ and $\tilde{J}_{N} \xrightarrow{\text{P}} J$. Furthermore, since $p_1,p_2,\ldots$ are i.i.d., due to the strong law of large numbers, we can verify that $\bar{J}_{N}$ converges to $J$ almost surely. We can further derive the rate of convergence (in probability) for a special case as discussed below. Suppose the state-transition and cost functions are linear (motivated by the fact that the linear models capture the state dynamics well in most control problems) as described below:
\begin{equation}\label{eq:linstate}
\begin{aligned}
x_{k+1} &= A x_k + B u_k + w_k, \, w_k \sim \mathcal{N}(0,\Sigma)\\
g(x_k,u_k) &= C x_k + D u_k,
\end{aligned}
\end{equation}
where $g(x_k,u_k)$ is a scalar function. The cost from the sample trajectory $i$ is given by
\begin{equation}
p_i = \sum_{k=1}^H g(x_k^i,u_k) = \sum_{k=1}^H (Cx_k^i + Du_k),
\end{equation} 
where $x_k^i$ is the sampled state at time step $k$ from the $i$th trajectory. Using the linear expressions in Eq.~\ref{eq:linstate}, we can verify $p_i$ further satisfies the following equation:
\begin{equation}\label{eq:pEp}
p_i - \text{E}[p_i] = C\left[ \sum_{k=0}^{H-1} \left(\sum_{q=0}^{H-k-1} A^q\right)w_k \right] = C\left[ \sum_{k=0}^{H-1} \mathcal{A}_k w_k \right],  
\end{equation}
where $\mathcal{A}_k = \sum_{q=0}^{H-k-1} A^q$. 
\begin{prop}
For a given sequence of actions $u_0, \ldots, u_{H-1}$
\begin{equation}
\text{P}\left(|J_N - J| \geq \epsilon \right) \leq \frac{\text{constant}}{N\epsilon^2}. 
\end{equation}    
\end{prop}

\begin{proof}
Let $p$ represent the cost for a sampled state trajectory. Using Eq.~\ref{eq:pEp}, we can verify 
\begin{equation}
\begin{aligned}
\text{Var}(p) &= \text{E}\left[(p - \text{E}[p])^{\text{T}} (p - \text{E}[p]) \right]\\
&= C\left[ \sum_{k=0}^{H-1} \mathcal{A}_k \Sigma \mathcal{A}_k^{\text{T}} \right] C^{\text{T}},
\end{aligned}
\end{equation}
which is a real scalar. Thus, $\text{Var}(J_N) = \text{Var}(p)/N$. 

Using Chebyshev's inequality, we can verify easily that
\begin{equation}
\text{P}\left(|J_N - J| \geq \epsilon \right) \leq \frac{\text{Var}(p)}{N\epsilon^2} = \frac{C\left[ \sum_{k=0}^{H-1} \mathcal{A}_k \Sigma \mathcal{A}_k^{\text{T}} \right] C^{\text{T}}}{N\epsilon^2}. 
\end{equation}
Furthermore, 
\begin{equation}
\lim_{N\rightarrow \infty} \text{P}\left(|J_N - J| \geq \epsilon \right) = 0, 
\end{equation}
which shows the convergence in probability as well. 
\end{proof}

\section{Case Studies}
We implement the above-discussed MHP methods in the context of two case studies: (a) linear quadratic Gaussian control (LQG); (b) path planning for unmanned aerial vehicles (UAVs). These case studies are discussed below.  
\subsection{Linear Quadratic Problem}
Although there are closed-form solutions for LQG problems, the below example allows us to quantify the benefits of using RS-MHP methods over existing similar methods, particularly NBO. Let the system state evolve according to the following linear equation:
\begin{equation}x_{k+1} = (1-a)x_k + au_k + w_k, \quad w_k \sim \mathcal{N}(0,\sigma^2),\end{equation}
where $0<a<1$ is a constant, and $w_k$ is a random disturbance modeled by a zero-mean Gaussian distribution with variance $\sigma^2$. The cost function over the time-horizon $H$ is defined as follows:
\begin{equation}J = \text{E} \left[ r(x_H-T)^2 + \sum_{k=0}^{H-1} u_k^2 \right],\end{equation}
where $r$ and $T$ are constants. This is a simplified oven temperature control example borrowed from \cite{bertsekas_lect}.

If we apply the traditional NBO method, assuming $H=2$, the cost function $J$ is approximated (assuming nominal values or zeros for $w_0$ and $w_1$) as 
\begin{equation}J_{\text{NBO}} = r\left((1-a)^2 x_0 + a(1-a)u_0 + a u_1 - T\right)^2 + u_0^2 + u_1^2\end{equation}
and the exact cost function $J$ can be evaluated analytically as
\begin{equation}\begin{aligned} 
J &= r\left((1-a)^2 x_0 + a(1-a)u_0 + a u_1 - T\right)^2 + u_0^2 + u_1^2\\
&+ r\sigma^2\left((1-a)^2 + 1\right).  
\end{aligned}
\end{equation}
We notice the approximation error due to the NBO method is \begin{equation}|J_{\text{NBO}} - J| = r\sigma^2\left((1-a)^2 + 1\right).\end{equation}
This approximation error for a generic time-horizon $H$ (the above error term is derived for $H=2$) is given by
\begin{equation}|J_{\text{NBO}} - J| = r\sigma^2\sum_{n=0}^{H-1}(1-a)^{2n}.\end{equation}
The above expression suggests that the NBO approximation error can be significantly high depending on the parameters $a$, $\sigma$, and $r$. With MHP approximation, the cost function reduces to
\begin{equation}J_{\text{MHP}} = \frac{1}{P} \left(\sum_{i=1}^P r(x_H^i-T)^2\right) + \sum_{k=0}^{H-1} u_k^2,\end{equation}
where $P$ is the number of state-trajectories generated using the MHP approach, and $x_H^i$ is the final state in the $i$th trajectory. Lemma~\ref{lemma:stronglaw} shows that the approximation error due to the above MHP method converges (in probability) to zero. We verify this result with a numerical simulation, where we implement the NBO and the MHP methods with the following assumptions: $x_0 = 0, r=10, T=1, H=2, u_0 = 0.55, u_1 = 0.17, \sigma = 1$. We vary $P$ from 100 to 10000 with increments of 100. Figure~\ref{fig:mhp_nbo} shows the cost function approximated using MHP and NBO methods. The figure clearly demonstrates that the error due to NBO approximation can be significantly high, while MHP performs better in cost approximation.
\begin{figure}[h]
\centering
    \includegraphics[width= \columnwidth, trim = 100 230 100 260,clip]{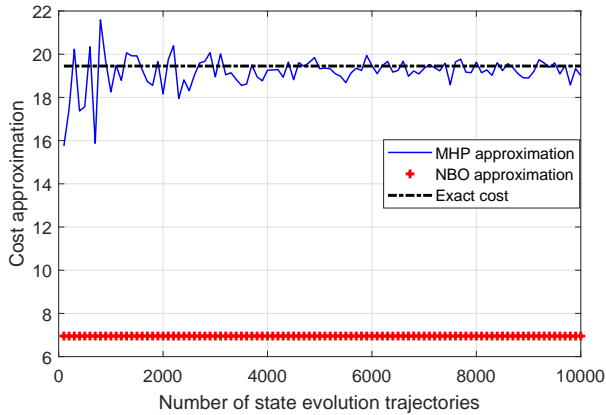}
    \caption{LQG problem: MHP vs.\ NBO}
    \label{fig:mhp_nbo}
\end{figure}

\subsection{UAV path planning problem}
We consider a UAV path planning problem, where the goal is to optimize the kinematic controls of a UAV to maximize a target tracking performance measure. Here, the UAV is assumed to be equipped with a sensor on-board that generates the location measurements of the target (a ground-based moving vehicle) corrupted by random noise. A detailed description of the problem can be found in \cite{Ragi_UAVpathPlanning}. In \cite{Ragi_UAVpathPlanning}, we posed this problem as a \emph{partially observable Markov decision process} (POMDP), where the POMDP led to solving a long horizon optimal control problem. We applied the NBO approach to solve the above POMDP. The resulting UAV path optimization problem is summarized as follows: 
\begin{equation*}
\min_{u} \, \text{E}\left[ \sum_{k=0}^{H-1} \text{tr} \,(\mathbf{P}_k(u)) \right] \xrightarrow{\text{NBO approx.}} \min_{u} \, \sum_{k=0}^{H-1} \text{tr} \,(\mathbf{\hat P}_k(u)),
\end{equation*}
where $\mathbf{P}_k(u)$ (a random variable) represents the error co-variance  matrix corresponding to the state of the system, tr() represents the matrix trace operator, $u$ is the sequence of UAV kinematic controls (e.g., forward acceleration and bank angle) applied over the discrete time planning horizon of length $H$ steps. After NBO approximation, the expectation over the random evolution of $\mathbf{P}_k(u)$ is replaced with the nominal sequence of the state covariance matrices $\text{tr} \,(\mathbf{\hat P}_k(u))$.

We now approximate the above objective function using the RS-MHP approach as follows:
\begin{equation*}
\begin{aligned}
\min_{u} \, \text{E}\left[ \sum_{k=0}^{H-1} \text{tr} \,(\mathbf{P}_k(u)) \right] &\xrightarrow{\text{RS-MHP approx.}} \\
&\min_{u} \, \frac{1}{N_T}\sum_{i=1}^{N}\sum_{k=0}^{H-1} \text{tr} \,(\mathbf{\tilde P}_k^i(u)),
\end{aligned}
\end{equation*}   
where $\mathbf{\tilde P}_k^i$ represents the state covariance matrix obtained from the $i$th state trajectory generated from the RS-MHP approach, and $N_T$ is the number of state trajectories. We implement this RS-MHP approach in MATLAB and run a Monte-Carlo study to see the impact of $N_T$ on the performance of the above UAV path planning algorithm, which is measured by the average target location estimation error. Figure~\ref{fig:MHP_MC} shows the cumulative distribution of average target location estimation errors from the RS-MHP approach with $H=6$, and for $N_T$ set to 50, 100, and 250. The figure shows a gradual increase in the UAV path optimization performance with increasing $N_T$ as expected. This result, as expected, also suggests that pruning methods (discussed in the previous section) would degrade the performance of the RS-MHP methods but can provide gains in terms of computational intensity.         
      
\begin{figure}[h]
\centering
    \includegraphics[width= \columnwidth, trim = 10 200 10 190,clip]{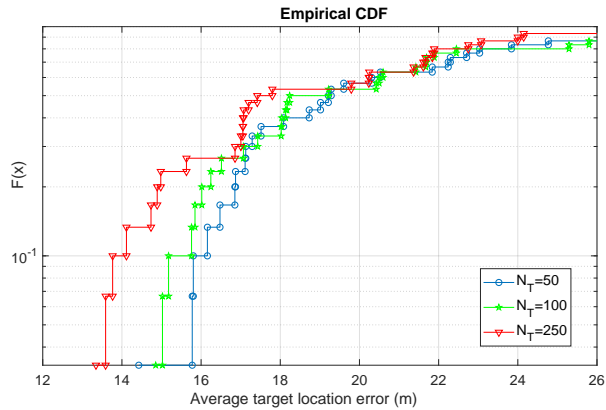}
    \caption{Cumulative distribution of average target location errors. Here $N_T$ represents the number of state evolution trajectories.}
    \label{fig:MHP_MC}
\end{figure}

RS-MHP has better capability in approximating the expectation operator in Eq.~\ref{eq:lhc_cost} than the NBO approach as we consider multiple hypotheses of state trajectories in RS-MHP as opposed to a single hypothesis in NBO as demonstrated in Figure~\ref{fig:RS_MHP_vs_NBO}. This is demonstrated in the above case studies.       

\begin{figure}[h]
\centering
    \includegraphics[width= \columnwidth, trim = 70 225 50 220,clip]{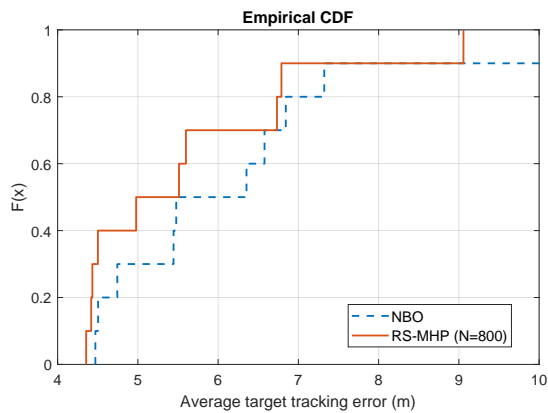}
    \caption{Cumulative distribution of average target location errors: NBO vs.\ RS-MHP.}
    \label{fig:RS_MHP_vs_NBO}
\end{figure}

\section{Conclusions}
In this paper, we developed a Monte-Carlo tree search method called \emph{random sampling - multipath hypothesis propagation} or RS-MHP to approximate the expectation operator in long horizon optimal control problems. Although variants of these methods exist in the literature, we focused on the convergence analysis of these approximation methods. The basic theme of these methods is to evolve the system state over multiple trajectories into the future while sampling the noise disturbances at each time-step. We derive convergence results that show that the cost approximation errors from our RS-MHP methods converge (in probability) toward zero as the sample size increases. We conducted a numerical study to assess the performance of our methods in two case studies: linear quadratic control problem and UAV path optimization problem. In both case studies, we demonstrated the benefits of our approach against an existing approach called \emph{nominal belief-state optimization} or NBO (used as a benchmark). 
       
\section{Acknowledgment}
The authors would like to thank Nicolas Lanchier, Arizona State University, for his valuable inputs and feedback on the convergence results discussed in this paper.

\bibliographystyle{IEEEtran} 
\bibliography{ref}

\end{document}